\documentclass[journal]{IEEEtran}

\usepackage{amsmath,amssymb,amsfonts,amsthm,mathtools}
\usepackage{cite}
\usepackage{booktabs}
\usepackage{array}
\usepackage{graphicx}
\usepackage{microtype}
\usepackage{xcolor}
\usepackage[caption=false,font=footnotesize]{subfig}
\usepackage[colorlinks=true,linkcolor=blue,citecolor=blue,urlcolor=blue]{hyperref}

\newtheorem{theorem}{Theorem}

\newtheorem{proposition}[theorem]{Proposition}
\newtheorem{corollary}[theorem]{Corollary}
\newtheorem{definition}[theorem]{Definition}
\newtheorem{assumption}[theorem]{Assumption}
\newtheorem{remark}[theorem]{Remark}

\newif\ifredline
\redlinefalse

\newcommand{\R}{\mathbb R}

\newcommand{\Pp}{\mathbb P}
\newcommand{\Ee}{\mathbb E}
\newcommand{\one}{\mathbf 1}

\newcommand{\cB}{\mathcal B}

\newcommand{\cE}{\mathcal E}
\newcommand{\cF}{\mathcal F}
\newcommand{\cG}{\mathcal G}
\newcommand{\cI}{\mathcal I}

\newcommand{\cS}{\mathcal S}

\newcommand{\cU}{\mathcal U}

\newcommand{\cW}{\mathcal W}
\newcommand{\cX}{\mathcal X}

\newcommand{\eps}{\varepsilon}
\newcommand{\bbeta}{\beta}

\DeclareMathOperator{\rad}{rad}

\title{Collective Tube Model Predictive Control With Distribution-Free Joint Safety Certificates}

\author{Giuseppe C. Calafiore, {\em Fellow,} IEEE
\thanks{Department of Electronics and Telecommunications, Politecnico di Torino, Italy. \tt{giuseppe.calafiore@polito.it}}}

\begin{document}
\maketitle

\begin{abstract}
Data-calibrated stochastic MPC typically builds separate risk margins for many events along the horizon, such as times, facets, state/input components,  or obstacles, and then combines them with a union bound.  
This approach is valid, but it does not match the key object used in the tube-MPC recursive-feasibility proof, which shifts a complete error tube.  This paper develops collective tube MPC (CT-MPC), where the calibrated uncertainty object is the finite-horizon prediction-error trajectory.  
A reusable trajectory tube is calibrated offline, its cross-sections define deterministic Pontryagin tightenings online, and the certified violation event is that a fresh prediction-error trajectory leaves the tube.
For linear systems with additive uncertainty and fixed ancillary feedback, we prove joint state-input safety over the prediction horizon, one-step recursive feasibility from an explicit shifted candidate, a finite-deployment risk bound, and a practical value-decrease inequality.  
The finite-sample certificate is distribution-free under split calibration and has beta-binomial form; its complexity is the certified number of residual trajectories that can define the tube, rather than the number of horizon-constraint blocks.  
We also give implementable shift-compatibility tests for polytopic tubes and a stable-compression fallback for irregular tube designers.  Numerical experiments compare CT-MPC with Bonferroni tightening, a sample-envelope tube, and a joint-in-time conformal MPC baseline.  The collective tube reduces deterministic tightening while preserving empirical safety and recursive feasibility. 

\end{abstract}

\begin{IEEEkeywords}
Model predictive control, stochastic MPC, tube MPC, recursive feasibility, chance constraints, scenario optimization, conformal prediction, distribution-free control.
\end{IEEEkeywords}

\section{Introduction}

Model predictive control (MPC) handles constraints by repeatedly solving a finite-horizon optimal-control problem, \cite{Mayne2000,Rawlings2020}.  In robust and stochastic MPC, tube constructions separate nominal motion from error motion: the nominal trajectory is optimized online, the error trajectory is confined to a tube, and recursive feasibility is proved by shifting the previous nominal solution and appending a terminal controller, see, e.g., \cite{Langson2004,Mayne2005,Cannon2011,Farina2016,Hewing2020Automatica}.  
Learning-based control adds a statistical design layer to this construction.  Scenario MPC, randomized chance-constrained MPC, probabilistic reachable sets, and conformal prediction provide finite-sample ways of replacing an unknown disturbance law by certified uncertainty sets, \cite{CalafioreFagiano2013,Schildbach2014,CampiPrandini2019,Hewing2020Lcss,Vovk2005,ShaferVovk2008,AngelopoulosBates2023}.

Recent conformal-control and conformal-MPC results show that distribution-free calibration can be part of a control proof, including trajectory-level confidence regions, dynamic-obstacle problems, and joint-in-time chance constraints for stochastic linear systems, see, e.g., \cite{Lindemann2025,Stamouli2024,Vlahakis2024,Vogel2026,Wang2025}.  The issue addressed here is the match between the statistical object and the control object.  A common implementation learns margins for a list of events, such as  times, facets, components, obstacles, or agents, and combines their risks by Boole's inequality.  This list-of-events view is natural statistically, but tube MPC shifts a complete error trajectory and its cross-sections.  When the residual trajectory is strongly dependent across time and constraints, marginal allocation can discard useful dependence information and can make the deterministic MPC problem tighter than necessary.

In CT-MPC we focus on the tube itself as the primitive statistical object.  The calibration data are complete residual trajectories, and  an offline design selects one trajectory tube, possibly after discarding complete residual trajectories, subject to deterministic shift-compatibility and terminal-set conditions.  The online MPC problem is then the standard tightened tube-MPC problem induced by the tube projections.  The certificate is attached to the single event that a fresh residual trajectory leaves the calibrated tube.  Thus, the probability statement, the tightening geometry, and the recursive-feasibility proof refer to the same object.

The viewpoint is complementary to randomized stochastic MPC.  Scenario MPC and randomized chance-constrained MPC can also certify joint events without allocating risk to every constraint block, and their sample complexity is governed by support rank, compression, or decision complexity rather than by the raw number of constraints, see \cite{CalafioreFagiano2013,Schildbach2014,CampiPrandini2019,Margellos2014,CampiGaratti2008}.  
The distinction, however, is architectural: rather than reinventing the recursive feasibility machinery or requiring complex online randomized solvers, our proposed approach preserves the classical deterministic tube-MPC law and its standard shifted-candidate proofs. Randomization is strictly confined to the offline certification of a reusable trajectory tube. This is a conscious design choice: by mapping the joint trajectory certificate directly into standard Pontryagin tightenings, the dependence-aware statistical design is paid for entirely offline, leaving a computationally lightweight, standard convex MPC problem for real-time execution.

Several recent works are especially close to the present approach.  Conformal SMPC methods, including \cite{Vogel2026}, construct finite-sample confidence regions for linear error trajectories and use them for joint-in-time chance-constrained MPC with recursive feasibility.  Conformal optimal-control work such as \cite{Vlahakis2024} also calibrates trajectory-level regions, and \cite{Stamouli2024} studies recursive feasibility for conformal MPC in dynamic environments.  The contribution of the present paper is to treat the calibrated trajectory set itself as a tube-design object: its projections determine deterministic tightenings, its support/compression complexity determines the finite-sample certificate, and its shift compatibility determines whether the usual tube-MPC recursive-feasibility argument survives calibration.

To clarify the scope,  we underline that the term \emph{distribution-free} means that no parametric distribution, Gaussian approximation, or independence assumption within a residual trajectory is imposed, but this does not remove the need for a representative calibration protocol.  The finite-sample certificate is conditional on the model class, ancillary gain, tube family, score geometry, and all tuning choices having been fixed before the calibration split is used. 
 If residual laws vary with unmodeled covariates or closed-loop operating regimes, then the calibration split must be updated online, or replaced by an explicitly covariate-conditional certificate.  Our key result certifies one-horizon tube containment and, by a union bound across MPC updates, a finite-deployment guarantee; however, it does not provide an infinite-horizon hard-safety certificate, unless a deterministic robust backup is added.

\vspace{.1cm}
\noindent
The main contributions are as follows:
\begin{enumerate}
\item We introduce collective tube MPC, a receding-horizon architecture in which the calibrated uncertainty object is the complete finite-horizon prediction-error trajectory.  The tube projections are used for Pontryagin tightening, while the probability certificate is attached to one joint tube-containment event.
\item We prove that CT-MPC preserves the standard control-theoretic structure of tube MPC for linear additive systems with fixed ancillary feedback.  Under explicit shift-compatibility and terminal inclusions, tube containment implies joint state-input safety over the prediction horizon, one-step recursive feasibility, a finite-deployment failure bound, and a practical value-function decrease.
\item We develop a finite-sample certification theory for collective tubes. Proper collective tube calibrators satisfy a beta-binomial risk certificate; for more general tube designers, the relevant quantity is a certified support/compression complexity rather than the number of time-constraint blocks or the observed active set in one data realization. Crucially, our formulation relies purely on set-theoretic boundary conditions rather than convex support constraints, completely decoupling the statistical certificate from the geometry of the deterministic MPC problem.

\item We give implementable synthesis conditions for shift-compatible polytopic tubes, prove an optimization dominance statement over modular discard allocation for shared monotone objectives, and compare CT-MPC with Bonferroni, sample-envelope, and joint conformal MPC baselines in a fully reproducible numerical study.
\end{enumerate}

\vspace{.1cm}
We adopt  linear dynamics in the main development because, in such case, under additive uncertainty and fixed ancillary feedback, the nominal and error dynamics separate exactly. This separation 
yields trajectory-independent calibrated tubes, explicit Pontryagin tightenings, and a clean recursive-feasibility proof. 
This choice thus isolates  the statistical contribution of the paper from additional nonlinear tube-design issues, and it leads to convex MPC problems. 
We remark, however, that  several ingredients of the proposed approach do not rely on linearity. In particular, the offline calibration of closed-loop residual trajectories, the beta-binomial certificate, the use of a reusable collective tube, and the resulting sample-based safety interpretation extend directly to any setting in which trajectories can be simulated and a nominal feedback policy is fixed. 
What becomes substantially harder in the nonlinear case is the deterministic MPC layer built around the calibrated tube, \cite{Mayne2011Nonlinear,Kohler2021Nonlinear}. The prediction error may in fact depend strongly on the nominal state and input, the propagated tube can be nonconvex and trajectory-dependent, and recursive feasibility or stability requires additional assumptions such as incremental stability, contraction, or computable nonlinear robust invariant sets. A full nonlinear theory combining the present collective calibration idea with such nonlinear tube constructions is therefore left as an open problem.

The paper is organized as follows.  Section~\ref{sec:problem} formulates the stochastic tube-MPC problem and the sampling model.  Section~\ref{sec:collective_tubes} defines collective tubes and proves the finite-sample tube certificate.  Section~\ref{sec:mpc} presents the CT-MPC problem.  Section~\ref{sec:rf} proves joint safety, recursive feasibility, finite-deployment certification, and practical stability.  Section~\ref{sec:allocation} compares collective and modular allocation.  Section~\ref{sec:synthesis} gives implementable synthesis conditions.  Section~\ref{sec:numerics} reports the numerical study, and Section~\ref{sec:conclusion} draws conclusions.

\section{Problem Setup}\label{sec:problem}
\subsection{System, constraints, and feedback decomposition}
We consider the uncertain linear system
\begin{equation}
    x_{k+1}=Ax_k+Bu_k+w_k,
    \label{eq:system}
\end{equation}
where $x_k\in\R^n$, $u_k\in\R^m$, and $w_k$ is an unknown disturbance. The state and input constraints are compact convex sets
\begin{equation}
    x_k\in\cX,
    \qquad
    u_k\in\cU,
    \label{eq:constraints}
\end{equation}
with the origin in their interiors. The distribution of the disturbance sequence is not assumed known and may be non-Gaussian, unbounded, and correlated in time. The only statistical information used for calibration is a collection of residual or disturbance trajectories generated under the same prediction architecture used online.

Let $K$ be a prestabilizing ancillary gain and set
\begin{equation}
    A_K:=A+BK.
\end{equation}
For a nominal state-input sequence $(z_t,v_t)$ satisfying
\begin{equation}
    z_{t+1}=Az_t+Bv_t,
    \label{eq:nominal_dyn}
\end{equation}
we write the actual prediction as $x_t=z_t+e_t$ and use the affine feedback
\begin{equation}
    u_t=v_t+Ke_t.
    \label{eq:ancillary_feedback}
\end{equation}
The prediction error evolves as
\begin{equation}
    e_{t+1}=A_Ke_t+w_t.
    \label{eq:error_dyn}
\end{equation}
At the beginning of each MPC optimization epoch we use measured-state initialization, so $z_0=x_k$ and $e_0=0$. The finite-horizon error trajectory is
\begin{equation}
    \mathbf e=(e_1,\ldots,e_H)\in\R^{nH}.
\end{equation}
The control objective is to solve a receding-horizon problem that satisfies a joint-in-time chance requirement of the form
\begin{equation}
\Pp\left\{
\begin{array}{l}
 x_{t\mid k}\in\cX,\ t=1,\ldots,H,\\
 u_{t\mid k}\in\cU,\ t=0,\ldots,H-1
\end{array}
\middle|\cI_k\right\}
\ge 1-\eps,
\label{eq:joint_chance}
\end{equation}
where $\cI_k$ is the information available at the planning time. Requirement \eqref{eq:joint_chance} is a single trajectory event, which differs from  imposing a separate chance constraint at each time and summing the risks.

\subsection{Trajectory samples and conditional sampling model}
The calibration data are residual trajectories
\begin{equation}
    \mathbf E^i=(E_1^i,\ldots,E_H^i),\qquad i=1,\ldots,N,
    \label{eq:calib_traj}
\end{equation}
obtained by propagating disturbance realizations through \eqref{eq:error_dyn}, by using held-out rollout residuals of a learned predictor, or by replaying stored episodes under the same information pattern. The fresh online residual trajectory at a planning time is denoted by $\mathbf E$.

\begin{assumption}[Calibration split and conditional trajectory sampling]\label{ass:exchangeability}
{The model class, the ancillary gain $K$, the tube family, and every tuning choice that affects the score, tube geometry, or support-complexity calculation are fixed before the calibration residuals \eqref{eq:calib_traj} are used. Equivalently, such choices are made on a separate training/tuning split and then frozen, or their own data dependence is included explicitly in the certified compression/support complexity. Conditional on this training/tuning information and on the online information pattern under which residuals are generated, the calibration residual trajectories $\mathbf E^1,\ldots,\mathbf E^N$ and any finite number of fresh residual trajectories are independent and identically distributed. In particular, every finite augmented collection of calibration and fresh trajectories is exchangeable.}
\end{assumption}

Assumption~\ref{ass:exchangeability} is the split-calibration sampling condition needed for the conditional violation-risk certificate $V_N$ in Theorem~\ref{thm:beta}. It is stronger than the minimal exchangeability condition used for one-shot marginal conformal coverage, but it is the standard condition behind conditional scenario and split-conformal risk bounds, 
\cite{BatesEtAl2021RiskControl,GarattiCampi2022}.
The assumption allows arbitrary temporal and spatial dependence \emph{within} each residual trajectory; the sampling requirement is across residual trajectories. In particular, errors at different prediction steps may be strongly correlated and heavy tailed. If the residual law depends on measured covariates, operating modes, or time-varying environments, then Assumption~\ref{ass:exchangeability} must be interpreted after conditioning on the chosen covariate stratum, or replaced by an appropriate weighted/Mondrian/online calibration theorem. Ties in boundary-defining scores are assumed to have probability zero or to be broken by an independent continuous tie-breaker; if deterministic ties can occur, the conservative certificate with a structural upper bound on the boundary size should be used.

\begin{remark}[Practicality of Exchangeability in Linear Error Dynamics]
While assuming identically distributed trajectories in general closed-loop operation is notoriously difficult due to state-dependent distribution shifts, the linear additive setting of \eqref{eq:system}--\eqref{eq:error_dyn} mitigates this. Because the error dynamics $e_{t+1} = A_K e_t + w_t$ separate completely from the nominal state $z_t$, the sequence $w_t$ drives the error independently of the current operating point. Consequently, a practitioner can construct valid, exchangeable calibration trajectories offline simply by replaying historical, unconditioned disturbance sequences through the fixed filter \eqref{eq:error_dyn}, perfectly simulating the error tubes the online controller will experience.
\end{remark}

\subsection{\texorpdfstring{{Limitations of modular allocation for MPC}}{Limitations of modular allocation for MPC}}
Suppose the joint event \eqref{eq:joint_chance} is decomposed into $M$ blocks, where a block may correspond to one time, one constraint facet, one output component, or one obstacle. If block $j$ is calibrated with violation risk $\eps_j$, then
\begin{equation}
    \Pp\{\hbox{some block violates}\}\le \sum_{j=1}^M\eps_j.
    \label{eq:bonferroni}
\end{equation}
This is the usual Bonferroni route. Its conservatism changes the control problem because the tightening is designed for marginal blocks rather than for the joint tube that is shifted in the MPC recursive-feasibility proof. In the following, 
we replace \eqref{eq:bonferroni} with a single calibrated tube event.

\section{Collective Trajectory Tubes}\label{sec:collective_tubes}
\subsection{Tube families and projections}
A collective trajectory tube is a measurable set
\begin{equation}
    \cE_\theta\subseteq\R^{nH},\qquad \theta\in\Theta,
    \label{eq:tube_family}
\end{equation}
where $\theta$ is a finite-dimensional design parameter. Its $t$-th cross-section is the projection
\begin{equation}
    \cS_t(\theta):=\{e_t:\exists(e_1,\ldots,e_H)\in\cE_\theta\},
    \qquad t=1,\ldots,H,
    \label{eq:projection}
\end{equation}
and $\cS_0(\theta):=\{0\}$. The online MPC problem uses the cross-sections for deterministic constraint tightening, while the probability certificate is assigned to the full event $\mathbf E\in\cE_\theta$.

The trajectory set itself need only be measurable for calibration.  For the deterministic MPC problem, however, we use compact convex cross-sections.  If the exact projection in \eqref{eq:projection} is nonconvex or difficult to represent, \(\cS_t(\theta)\) denotes any certified compact convex outer approximation of that projection.  This replacement can only tighten the nominal MPC constraints; it does not change the calibrated event \(\{\mathbf E\in\cE_\theta\}\), and it preserves the implication ``tube containment implies state-input safety.''
Examples include the following.
\begin{itemize}
\item \emph{Scalar trajectory-score tubes:} given a score $S:\R^{nH}\to\R$, set $\cE_q=\{\mathbf e:S(\mathbf e)\le q\}$. A typical MPC score is
\begin{equation}
    S(\mathbf e)=\max_{t,j}\frac{g_j^\top e_t}{\rho_{jt}},
    \label{eq:mpc_score}
\end{equation}
where $g_j$ are constraint normals and $\rho_{jt}>0$ are shape weights.
\item \emph{Rectangular trajectory tubes:} for facet functions $s_{jt}(e_t)$,
\begin{equation}
    \cE_q=\{\mathbf e:s_{jt}(e_t)\le q_{jt},\; j=1,\ldots,p,\; t=1,\ldots,H\}.
    \label{eq:rect_tube}
\end{equation}
Although \eqref{eq:rect_tube} is rectangular in score space, it is calibrated as one trajectory event.
\item \emph{Pareto and support-function tubes:} the tube may be the lower envelope of non-dominated residual trajectories, or a convex set described by finitely many support functions. These choices are useful when different constraints trade off along the horizon.
\end{itemize}

\subsection{Collective tube calibration}
Given calibration trajectories $S_N=(\mathbf E^1,\ldots,\mathbf E^N)$, a tube calibrator returns $\hat\theta_N=A_N(S_N)$. Its conditional trajectory violation risk is
\begin{equation}
    V_N:=\Pp\{\mathbf E\notin\cE_{\hat\theta_N}\mid S_N\}.
    \label{eq:VN}
\end{equation}
{The conditioning on the fixed training/tuning information and on the online information pattern in Assumption~\ref{ass:exchangeability} is suppressed in \eqref{eq:VN}.}
The tube can be selected by solving the offline design problem
\begin{equation}
\begin{aligned}
    \min_{\theta\in\Theta}\quad &\Phi_{\rm MPC}(\theta)\\
    \text{s.t.}\quad &\sum_{i=1}^N \one\{\mathbf E^i\notin\cE_\theta\}\le r,\\
    &\theta\in\Theta_{\rm sh},
\end{aligned}
\label{eq:collective_tube_design}
\end{equation}
where $\Phi_{\rm MPC}$ measures online conservatism, $r$ is a trajectory-discard budget, and $\Theta_{\rm sh}$ encodes the shift-compatibility conditions introduced in Section~\ref{sec:rf}. Examples of $\Phi_{\rm MPC}$ include the sum of support functions of the tightened sets, the loss of terminal-set volume, or the optimal value of a representative nominal MPC problem.

The important point is that the objects discarded in \eqref{eq:collective_tube_design} are complete trajectories. A residual trajectory that is extreme at many times consumes one discard, rather than many independent risk budgets.

The discard budget \(r\) is not, by itself, the statistical complexity of an arbitrary optimized tube.  For a scalar order-statistic score, the complexity is exactly \(r+1\).  For a multi-parameter shape optimization, the certificate must also account for the residual trajectories needed to reconstruct the active tube shape, and for any data-dependent score or template choice.  Thus the number inserted into Theorem~\ref{thm:beta} is a certified support/compression bound for the whole tube-design rule, not merely the number of discarded calibration trajectories.  This distinction prevents the most common misuse of scenario and conformal certificates in data-designed MPC.

\subsection{Finite-sample certificate for a calibrated tube}
For $s\in\{1,\ldots,N\}$ define
\begin{equation}
    B_{N,s}(\eps):=
    \sum_{j=0}^{s-1}\binom{N}{j}\eps^j(1-\eps)^{N-j},
    \label{eq:BNs}
\end{equation}
and
\begin{equation}
    \eps_{N,s}(\bbeta):=\inf\{\eps\in[0,1]:B_{N,s}(\eps)\le\bbeta\}.
    \label{eq:epsNs}
\end{equation}
It is convenient to name the calibration-success event
\begin{equation}
    \cG_N(\bbeta):=\{V_N\le \eps_{N,s}(\bbeta)\}.
    \label{eq:calibration_event}
\end{equation}
Theorem~\ref{thm:beta} states that \(\Pp^N\{\cG_N(\bbeta)\}\ge 1-\bbeta\).  After calibration, \(\cG_N(\bbeta)\) is not directly observable; it is a PAC statement about the random calibration sample.  All closed-loop probability bounds below are to be read conditionally on this event, and hence hold with the same calibration confidence.

The next result is stated for tube calibrators that have a proper boundary. It is a compact version of the scenario/conformal argument, included here because it is the statistical certificate used by CT-MPC.

For a deterministic collection $T_n=(\mathbf e^1,\ldots,\mathbf e^n)$ and an index set $I\subset[n]$, let $A(T_I)$ be the tube parameter returned by the calibrator trained on $T_I$, and write $\cE(T_I):=\cE_{A(T_I)}$. A boundary map assigns a set $\cB_n(T_n)\subset[n]$.

\begin{definition}[Proper collective tube boundary]\label{def:proper}
A calibrator is a {\em proper collective tube boundary scheme} with size $s$ if, outside null tie events, for every $n\ge s$, every deterministic $T_n$, and every split {$I\subset[n]$ with $|I|\ge s$}, $J=[n]\setminus I$, the following hold:
\begin{enumerate}
\item Boundary equivalence:
\begin{equation}
    \mathbf e^j\in\cE(T_I)\;\forall j\in J
    \quad\Longleftrightarrow\quad
    \cB_n(T_n)\subseteq I.
    \label{eq:boundary_equiv}
\end{equation}
\item Projectivity: if $\cB_n(T_n)\subseteq I$, then $\cB_{|I|}(T_I)=\cB_n(T_n)$ after relabeling.
\item 
Fixed boundary size: \({|\cB_n(T_n)|=s}\).
\end{enumerate}
If items 1)--2) hold and item 3) is replaced by
\[
{|\cB_n(T_n)|\le \bar s}
\]
for all \(n\) and all deterministic samples \(T_n\), then we say that the
calibrator has certified boundary complexity at most \(\bar s\).
The number \(\bar s\) is a structural upper bound of the calibration rule.
\end{definition}

\begin{remark}[Departure from convex scenario optimization]\label{rem:departure_scenario}
While the resulting beta-binomial certificate shares the form familiar from the classical scenario approach \cite{CampiGaratti2008}, Definition~\ref{def:proper} abstracts the complexity mechanism entirely from convexity. In classical scenario optimization, the support size is strictly tied to the number of support constraints (related to Helly's dimension) of a convex program. By requiring only the combinatorial properties of boundary equivalence and projectivity, our framework allows for non-convex, discrete, or nonlinear tube calibration rules (such as scalar trajectory-score thresholds or Pareto frontiers) while preserving the exact finite-sample certificate.
\end{remark}

\begin{remark}[A concrete  fixed-\(s\) example.] Let \(S(\mathbf e)\) be any scalar score of the complete prediction-error trajectory, for instance the maximum normalized facet violation in \eqref{eq:mpc_score}. Fix an integer discard budget \(r\) before calibration and set \(q_N\) equal to the \((r+1)\)-st largest value among \(S(\mathbf E^1),\ldots,S(\mathbf E^N)\). The calibrated tube is \(\cE_{q_N}=\{\mathbf e:S(\mathbf e)\le q_N\}\), so exactly the \(r\) largest-score calibration trajectories may be excluded. Outside ties, the boundary of any deterministic sample is the set of its \(r+1\) largest-score trajectories; retaining this set preserves the same threshold and accepts every omitted trajectory, while omitting one of these trajectories lowers the threshold or leaves an omitted trajectory outside. Hence the proper-boundary conditions hold with the fixed collective support size \(s=r+1\). The case \(r=0\) is the sample-envelope tube with \(s=1\). Importantly, a boundary trajectory may be extreme at many times or facets, but it still contributes one unit to \(s\), because the certified event is the whole-trajectory event.
\end{remark}

\begin{theorem}[Collective tube risk certificate]\label{thm:beta}
{Let $\mathbf E^1,\ldots,\mathbf E^N$ be calibration residual trajectories satisfying Assumption~\ref{ass:exchangeability}.}

{\emph{(i)}} If the tube calibrator is a proper collective tube boundary scheme with boundary size $s$, then
\begin{equation}
    V_N\sim {\rm Beta}(s,N-s+1).
    \label{eq:beta_law}
\end{equation}
Equivalently, for every $\eps\in[0,1]$,
\begin{equation}
    \Pp^N\{V_N>\eps\}=B_{N,s}(\eps).
    \label{eq:pac_tail}
\end{equation}
Consequently,
\begin{equation}
    \Pp^N\{V_N\le \eps_{N,s}(\bbeta)\}\ge 1-\bbeta.
    \label{eq:pac_cert}
\end{equation}

{\emph{(ii)}} If the calibrator has certified
boundary complexity at most \(\bar s\), then the conservative certificate
\[
    \Pp^N\{V_N>\eps\}\le B_{N,\bar s}(\eps)
\]
holds. Consequently,
\[
    \Pp^N\{V_N\le \eps_{N,\bar s}(\bbeta)\}\ge 1-\bbeta .
\]
\end{theorem}

\begin{proof}
{Draw $M$ fresh residual trajectories $\mathbf Y^1,\ldots,\mathbf Y^M$ from the same conditional law as the calibration trajectories, independently of $S_N$ given the training/tuning information and online information pattern. Conditionally on $S_N$,}
\begin{equation}
    \Pp\{\mathbf Y^1,\ldots,\mathbf Y^M\in\cE(S_N)\mid S_N\}=(1-V_N)^M.
    \label{eq:moment_start}
\end{equation}
Augment the sample to $T_{N+M}=(S_N,\mathbf Y^1,\ldots,\mathbf Y^M)$. {By the exchangeability of the augmented conditionally i.i.d. sample, training on the first $N$ samples is equivalent to first drawing $T_{N+M}$ and then selecting a uniformly random $N$-subset $I\subset[N+M]$ as the training indices.} Let $J=[N+M]\setminus I$. Boundary equivalence gives
\begin{equation}
    T_J\subset\cE(T_I)
    \quad\Longleftrightarrow\quad
    \cB_{N+M}(T_{N+M})\subseteq I.
\end{equation}
Conditional on $T_{N+M}$, the boundary is a fixed $s$-element set. The probability that a uniformly random $N$-subset contains it is
\begin{equation}
    \frac{\binom{N}{s}}{\binom{N+M}{s}}.
\end{equation}
Thus $\Ee[(1-V_N)^M]=\binom{N}{s}/\binom{N+M}{s}$ for all $M\ge0$. These are exactly the moments of $1-U$ when $U\sim {\rm Beta}(s,N-s+1)$. Moment determinacy on $[0,1]$ proves \eqref{eq:beta_law}; \eqref{eq:pac_tail} and \eqref{eq:pac_cert} follow from the beta-binomial identity.

For part (ii), the upper bound structurally mirrors the classical scenario approach \cite{CampiGaratti2008,CalafioreSIOPT2010}, where $\bar s$ traditionally represents a bound on the number of support constraints in a strictly convex program. However, because our abstract setting defined in Definition~\ref{def:proper} relies purely on the combinatorial boundary structure rather than convexity, the proof requires a different, more general lifting argument. We outline the proof idea here, but omit the exceedingly technical details: 
if the calibration rule has support complexity at most \(\bar s\), one may define a modified rule
in which  each boundary set of the original rule is padded with arbitrary non-boundary
indices until it has cardinality \(\bar s\).
Boundary equivalence and
projectivity can be proved to still hold for these padded boundary sets, hence 
the modified rule is a proper projective boundary scheme, now
with deterministic fixed boundary size \(\bar s\). The fixed boundary result on point (i), with exact beta law, then applies to the modified rule.
This,  associated with the fact that the padded rule is
more conservative than the original one (a stochastic dominance relation can be proved to hold), finally gives the stated beta-binomial upper bound on the probability of $V_N >\epsilon$.
\end{proof}

\begin{remark}[Role of the certificate]
Theorem~\ref{thm:beta} provides a statistical tube certificate. This
certificate is used in Theorem~\ref{thm:horizon_safety} for one-horizon
state/input safety and in Theorem~\ref{thm:rf} for recursive feasibility of the
resulting tube-MPC law. Corollary~\ref{cor:finite_deploy} combines these
one-step statements over a finite deployment. In the remainder of the paper,
\(s\) denotes either the exact fixed boundary size of a proper calibrator or a
conservative structural upper bound used in the beta-binomial certificate; it
is not the observed random support size of one calibration realization unless a
separate random-complexity theorem is invoked explicitly.
\end{remark}

\section{Collective Tube MPC}\label{sec:mpc}
\subsection{Online optimization problem}
Fix a calibrated tube parameter $\theta$ and let $\cS_t=\cS_t(\theta)$. For compact sets $A,B$, $A\ominus B:=\{a:a+B\subseteq A\}$ denotes the Pontryagin difference. The CT-MPC problem at state $x$ is
\begin{subequations}\label{eq:ctmpc}
\begin{align}
    J_H^\star(x;\theta):=
    \min_{\mathbf z,\mathbf v}\quad &
        \ell_f(z_H)+\sum_{t=0}^{H-1}\ell(z_t,v_t)\label{eq:ctmpc_cost}\\
    \text{s.t.}\quad
        &z_0=x,\label{eq:ctmpc_init}\\
        &z_{t+1}=Az_t+Bv_t,
        &&t=0,\ldots,H-1,\label{eq:ctmpc_dyn}\\
        &z_t\in\cX\ominus\cS_t,
        &&t=0,\ldots,H-1,\label{eq:ctmpc_state}\\
        &v_t\in\cU\ominus K\cS_t,
        &&t=0,\ldots,H-1,\label{eq:ctmpc_input}\\
        &z_H\in\cX_f(\theta).
        \label{eq:ctmpc_terminal}
\end{align}
\end{subequations}
Here $K\cS_t:=\{Ke:e\in\cS_t\}$. Since $\cS_0=\{0\}$, \eqref{eq:ctmpc_input} at $t=0$ is simply $v_0\in\cU$. The applied input is
\begin{equation}
    u_k=v_{0\mid k}^\star,
    \label{eq:apply_input}
\end{equation}
because the prediction error at the measured initial state is zero. The input tightening for $t\ge1$ is nevertheless essential and it is what makes the shifted candidate feasible after the first random error is realized.

\subsection{Deterministic admissibility of a tube}
The online problem \eqref{eq:ctmpc} is a deterministic optimization problem. The next definition specifies when its tightened constraints are compatible with the standard MPC shift.

\begin{definition}[Shift-compatible tube]\label{def:shift}
Let $A_K=A+BK$. A tube with cross-sections $\cS_0=\{0\},\cS_1,\ldots,\cS_H$ is shift-compatible if
\begin{equation}
    A_K^j\cS_1\oplus\cS_j\subseteq\cS_{j+1},
    \qquad j=0,\ldots,H-1.
    \label{eq:shift_compat}
\end{equation}
\end{definition}
Condition \eqref{eq:shift_compat} is the finite-horizon analogue of robust tube recursion. If $\cS_j=\bigoplus_{i=0}^{j-1}A_K^i\cW$ for a disturbance set $\cW$, then \eqref{eq:shift_compat} holds with equality. In the data-calibrated case, \eqref{eq:shift_compat} is imposed on the tube family $\Theta_{\rm sh}$ in \eqref{eq:collective_tube_design}.

Only the inclusions with \(j=0,\ldots,H-2\) are used to shift the stage tightenings in the proof of Theorem~\ref{thm:rf}; the inclusion with \(j=H-1\) is a convenient horizon-closure condition and is automatically satisfied by the standard reachable-set recursion.  If a less restrictive implementation omits it, the terminal inclusions in Assumption~\ref{ass:terminal} must be checked directly, as they are below.

\begin{assumption}[Terminal ingredients]\label{ass:terminal}
There exist a compact convex set $\cX_f(\theta)$ and a terminal gain $K_f$ such that, with $A_f:=A+BK_f$,
\begin{subequations}\label{eq:terminal_ass}
\begin{align}
    {\cX_f\oplus \cS_H} &{\subseteq \cX,}\label{eq:term_horizon_safe}\\
    \cX_f\oplus A_K^{H-1}\cS_1 &\subseteq \cX\ominus\cS_{H-1},\label{eq:term_state}\\
    K_f(\cX_f\oplus A_K^{H-1}\cS_1) &\subseteq \cU\ominus K\cS_{H-1},\label{eq:term_input}\\
    A_f(\cX_f\oplus A_K^{H-1}\cS_1) &\subseteq \cX_f.\label{eq:term_inv}
\end{align}
\end{subequations}
Moreover, the terminal cost satisfies the nominal Lyapunov inequality
\begin{equation}
    \ell_f(A_fz)-\ell_f(z)\le -\ell(z,K_fz),
    \qquad z\in\cX_f.
    \label{eq:terminal_cost}
\end{equation}
\end{assumption}
Assumption~\ref{ass:terminal} is a tube-MPC terminal condition. The additional set $A_K^{H-1}\cS_1$ appears because measured-state reinitialization shifts the previous nominal tail by the realized one-step prediction error.
The added inclusion \eqref{eq:term_horizon_safe} is the terminal safety tightening: it ensures that the terminal nominal state used in \eqref{eq:ctmpc_terminal} also satisfies the original state constraint after any terminal error in \(\cS_H\).

\section{Safety, Recursive Feasibility, and Stability}\label{sec:rf}
\subsection{Joint safety over one horizon}
The first theorem connects a certified trajectory tube to the joint chance constraint.

\begin{theorem}[Certified horizon safety]\label{thm:horizon_safety}
Suppose Assumption~\ref{ass:exchangeability} holds, $\theta$ is fixed after calibration, and \eqref{eq:ctmpc} is feasible at state $x_k$. If the feedback sequence $u_{t\mid k}=v_{t\mid k}^\star+Ke_{t\mid k}$ were applied over the prediction horizon, then
\begin{equation}
    \mathbf E_k\in\cE_\theta
    \quad\Longrightarrow\quad
    x_{t\mid k}\in\cX,\quad u_{t\mid k}\in\cU,
\end{equation}
for all $t=1,\ldots,H$ for the state constraints and all $t=0,\ldots,H-1$ for the input constraints. Consequently, if $\hat\theta_N$ satisfies \eqref{eq:pac_cert}, then with probability at least $1-\bbeta$ over the calibration data,
\begin{multline}
\Pp\{\text{some state or input constraint violates}\\
\text{over the horizon}\mid S_N\}
\le \eps_{N,s}(\bbeta).
\label{eq:horizon_certified}
\end{multline}
\end{theorem}

\begin{proof}
If $\mathbf E_k\in\cE_\theta$, then $e_{t\mid k}\in\cS_t$ for every $t$ by definition of the projections. For $t=1,\ldots,H-1$, $z_{t\mid k}^\star\in\cX\ominus\cS_t$ gives $z_{t\mid k}^\star+e_{t\mid k}\in\cX$. {For the terminal state, $z_{H\mid k}^\star\in\cX_f$ and \eqref{eq:term_horizon_safe} give $z_{H\mid k}^\star+e_{H\mid k}\in\cX$.} Similarly, $v_{t\mid k}^\star\in\cU\ominus K\cS_t$ implies $v_{t\mid k}^\star+Ke_{t\mid k}\in\cU$. {Thus a state/input violation can occur only on the tube-escape event \(\{\mathbf E_k\notin\cE_\theta\}\); its probability is therefore bounded by the tube violation risk \(V_N\).} Equation~\eqref{eq:horizon_certified} follows from Theorem~\ref{thm:beta}.
\end{proof}

\subsection{Recursive feasibility}
The next theorem is the central control result. It shows that collective calibration does not break the standard tube-MPC shift argument.

\begin{theorem}[Recursive feasibility on the calibrated tube]\label{thm:rf}
Let $\theta$ be such that the cross-sections of $\cE_\theta$ are shift-compatible in the sense of Definition~\ref{def:shift}, and suppose Assumption~\ref{ass:terminal} holds. If \eqref{eq:ctmpc} is feasible at time $k$ and the realized one-step prediction error satisfies
\begin{equation}
    e_{1\mid k}:=x_{k+1}-z_{1\mid k}^\star\in\cS_1,
    \label{eq:e1_good}
\end{equation}
then \eqref{eq:ctmpc} is feasible at time $k+1$.

If, in addition, $\hat\theta_N$ satisfies \eqref{eq:pac_cert}, then with calibration confidence at least $1-\bbeta$,
\begin{multline}
\Pp\{\text{CT-MPC is infeasible at }k+1\mid\\
\text{CT-MPC is feasible at }k,S_N\}
\le \eps_{N,s}(\bbeta).
\label{eq:rf_cert}
\end{multline}
\end{theorem}

\begin{proof}
Let $(z_0^\star,\ldots,z_H^\star,v_0^\star,\ldots,v_{H-1}^\star)$ be a feasible solution at time $k$. Define a candidate solution at time $k+1$ by
\begin{align}
    \tilde z_j &= z_{j+1}^\star + A_K^j e_{1\mid k},
    &&j=0,\ldots,H-1,\label{eq:shift_z}\\
    \tilde v_j &= v_{j+1}^\star + K A_K^j e_{1\mid k},
    &&j=0,\ldots,H-2,\label{eq:shift_v}\\
    \tilde v_{H-1} &= K_f\tilde z_{H-1},
    \label{eq:shift_terminal_input}\\
    {\tilde z_H} &{=A_f\tilde z_{H-1}.}
    \label{eq:shift_terminal_state}
\end{align}
The initial condition is satisfied because $\tilde z_0=z_1^\star+e_{1\mid k}=x_{k+1}$. For $j=0,\ldots,H-2$,
\begin{equation}
    A\tilde z_j+B\tilde v_j
    =z_{j+2}^\star+A_K^{j+1}e_{1\mid k}
    =\tilde z_{j+1}.
\end{equation}
{The last dynamics equation is
\(A\tilde z_{H-1}+B\tilde v_{H-1}=A_f\tilde z_{H-1}=\tilde z_H\).}
{For $j=0,\ldots,H-2$ and any $\eta\in\cS_j$,} shift compatibility gives $A_K^j e_{1\mid k}+\eta\in\cS_{j+1}$. Since $z_{j+1}^\star\in\cX\ominus\cS_{j+1}$,
\begin{equation}
    \tilde z_j+\eta=z_{j+1}^\star+A_K^j e_{1\mid k}+\eta\in\cX,
\end{equation}
which means $\tilde z_j\in\cX\ominus\cS_j$. The same argument, using $v_{j+1}^\star\in\cU\ominus K\cS_{j+1}$, gives $\tilde v_j\in\cU\ominus K\cS_j$ for $j=0,\ldots,H-2$. {For the last stage, \eqref{eq:term_state} gives \(\tilde z_{H-1}\in\cX\ominus\cS_{H-1}\), \eqref{eq:term_input} gives \(\tilde v_{H-1}\in\cU\ominus K\cS_{H-1}\), and \eqref{eq:term_inv} gives \(\tilde z_H\in\cX_f\).} Thus the shifted sequence is feasible at time $k+1$.

Finally, $\mathbf E_k\in\cE_\theta$ implies $e_{1\mid k}\in\cS_1$. Therefore infeasibility can occur only if $\mathbf E_k\notin\cE_\theta$, whose conditional probability is bounded by the certified risk.
\end{proof}

\begin{corollary}[Finite-deployment certificate]\label{cor:finite_deploy}

Fix a deployment horizon of \(T\) MPC updates and let \(\cF_k\) denote the closed-loop information available immediately before solving CT-MPC at time \(k\).  Suppose that, whenever the CT-MPC problem is feasible, the fresh horizon residual trajectory at update \(k\) satisfies the same calibrated conditional risk bound
\begin{equation}
    \Pp\{\mathbf E_k\notin\cE_{\hat\theta_N}\mid \cF_k,S_N\}
    \le \eps_{N,s}(\bbeta)
    \label{eq:conditional_update_risk}
\end{equation}
on the calibration event \(\cG_N(\bbeta)\).  Then, on \(\cG_N(\bbeta)\),
\begin{multline}
\Pp\{\text{some feasibility or horizon-safety failure}\\
\text{in }T\text{ updates}\mid S_N\}
\le T\eps_{N,s}(\bbeta).
\label{eq:T_bound}
\end{multline}
No independence between different MPC updates is required for \eqref{eq:T_bound}; overlap of prediction windows is allowed.  What is required is the per-update conditional risk bound \eqref{eq:conditional_update_risk}.  The factor \(T\) is a union bound across receding-horizon invocations, not across horizon steps or constraint facets.  Removing this remaining factor requires additional structure, such as online conformal updating, a martingale calibration argument,  or a deterministic robust backup.

\end{corollary}
\begin{proof}

Let \(F_k\) be the event that update \(k\) produces either a state/input horizon-safety failure or a one-step recursive-feasibility failure, conditional on feasibility before the update.  By Theorems~\ref{thm:horizon_safety} and \ref{thm:rf}, \(F_k\subseteq\{\mathbf E_k\notin\cE_{\hat\theta_N}\}\) on the deterministic CT-MPC inclusions.  Therefore \(\Pp(F_k\mid \cF_k,S_N)\le \eps_{N,s}(\bbeta)\) on \(\cG_N(\bbeta)\).  Summing \(\Pp(F_k\mid S_N)=\Ee[\Pp(F_k\mid \cF_k,S_N)\mid S_N]\) over \(k=0,\ldots,T-1\) gives \eqref{eq:T_bound}.

\end{proof}

A deterministic route to infinite-horizon hard safety would be to attach a robust backup tube to the calibrated controller. Let \(\bar\cW\) be a hard disturbance support set and suppose there are a backup feedback \(\kappa_b\) and a robust control-invariant safe set \(\cX_b\subseteq\cX\) such that
\begin{equation}
    \kappa_b(x)\in\cU,
    \qquad
    Ax+B\kappa_b(x)+\bar\cW\subseteq\cX_b,
    \qquad x\in\cX_b .
    \label{eq:backup_invariance}
\end{equation}
The CT-MPC problem can then be equipped with the one-step guard
\begin{equation}
    z_{1\mid k}\in \cX_b\ominus \bar\cW,
    \label{eq:backup_guard}
\end{equation}
which is affine in the nominal first input because \(z_{1\mid k}=Ax_k+Bv_{0\mid k}\). If the calibrated shift condition holds, the controller continues in the less conservative CT-MPC mode by Theorem~\ref{thm:rf}. If a tube-escape event is detected after measuring \(x_{k+1}\), the guard implies \(x_{k+1}\in\cX_b\) for every \(w_k\in\bar\cW\), and the controller may switch permanently, or until recertification, to \(\kappa_b\). The invariance \eqref{eq:backup_invariance} then gives safety for all subsequent times without accumulating the factor in \eqref{eq:T_bound}. 

In this hybrid implementation, the controller operates under a dual-tube paradigm: the tight, data-calibrated CT-MPC tube serves as the primary ``performance inner bound,'' dictating nominal operation with high probability, while the conservative robust backup tube serves purely as an ``infinite-horizon safety outer bound.'' This ensures the system does not pay the performance penalty of worst-case robust tightenings during nominal operation, but still possesses a mathematically guaranteed fail-safe if an extreme tail event forces a tube escape. If \(\bar\cW\) is itself only probabilistically calibrated, then this deterministic conclusion is replaced by the corresponding repeated-use certificate for the backup set.

\subsection{Practical stability on certified tube-containment intervals}
The next theorem records the value-function decrease inherited from tube MPC. It is stated in a practical form because measured-state reinitialization perturbs the shifted nominal tail by $e_{1\mid k}$. {Here \(\rad(S):=\sup_{s\in S}\|s\|\).}

\begin{assumption}[Lipschitz cost on the feasible tube]\label{ass:lipschitz}
The stage and terminal costs are continuous, positive definite, and locally Lipschitz on the compact feasible region induced by \eqref{eq:ctmpc}. {Moreover, the stage cost dominates the measured state: there exists a class-$\mathcal K_\infty$ function $\alpha_\ell$ such that $\ell(x,v)\ge \alpha_\ell(\|x\|)$ for all feasible $(x,v)$.} There are class-$\mathcal K_\infty$ functions $\alpha_1,\alpha_2$ such that
\begin{equation}
    \alpha_1(\|x\|)\le J_H^\star(x;\theta)\le \alpha_2(\|x\|)
    \label{eq:value_bounds}
\end{equation}
for every feasible $x$.
\end{assumption}

\begin{theorem}[Practical value decrease]\label{thm:stability}
Under the assumptions of Theorem~\ref{thm:rf} and Assumption~\ref{ass:lipschitz}, there exists a finite constant $L_\theta$ such that, whenever \eqref{eq:ctmpc} is feasible at $k$ and $e_{1\mid k}\in\cS_1$,
\begin{equation}
    J_H^\star(x_{k+1};\theta)-J_H^\star(x_k;\theta)
    \le -\ell(x_k,v_{0\mid k}^\star)+L_\theta\|e_{1\mid k}\|.
    \label{eq:value_decrease}
\end{equation}
{Consequently, if \(\|e_{1\mid k}\|\le \delta_\theta:=\rad(\cS_1)\) on an interval of certified tube containment, then
\[
    J_H^\star(x_{k+1};\theta)-J_H^\star(x_k;\theta)
    \le -\alpha_\ell(\|x_k\|)+L_\theta\delta_\theta .
\]
Together with \eqref{eq:value_bounds}, this gives the usual Lyapunov practical-stability estimate with ultimate radius no larger than
\(\alpha_1^{-1}\!\circ\alpha_2\!\circ\alpha_\ell^{-1}(L_\theta\delta_\theta)\), up to the standard compact-set constants. In particular, the bound shrinks as \(\rad(\cS_1)\) shrinks.}
\end{theorem}

\begin{proof}
Evaluate the cost of the shifted candidate constructed in \eqref{eq:shift_z}--\eqref{eq:shift_terminal_input}. If $e_{1\mid k}=0$, the standard terminal Lyapunov inequality \eqref{eq:terminal_cost} gives
\begin{equation}
    J_H^\star(x_{k+1};\theta)-J_H^\star(x_k;\theta)
    \le -\ell(x_k,v_{0\mid k}^\star).
\end{equation}
For nonzero $e_{1\mid k}$, every shifted state and input differs from the zero-error shifted candidate by a linear image of $e_{1\mid k}$. Since all feasible candidates lie in a compact set and the costs are Lipschitz there, the additional cost is bounded by $L_\theta\|e_{1\mid k}\|$ for a finite constant $L_\theta$. Optimality of $J_H^\star(x_{k+1};\theta)$ gives \eqref{eq:value_decrease}. {The practical-stability conclusion follows by combining the last inequality with the stage-cost lower bound, \eqref{eq:value_bounds}, and the standard Lyapunov argument for difference inequalities with a bounded additive term.}
\end{proof}

The theorem separates the two roles of the calibrated tube. Its cross-sections must be small enough for a meaningful practical stability bound, while its trajectory event must have high enough probability for the stability statement to apply frequently.

\section{Collective Allocation Versus Modular Allocation}\label{sec:allocation}
\subsection{MPC tightening objective}
For polyhedral constraints
\begin{equation}
    \cX=\{x:F_xx\le b_x\},\qquad
    \cU=\{u:F_uu\le b_u\},
\end{equation}
constraint tightening by cross-sections $\cS_t$ is equivalent to replacing the right-hand sides by
\begin{align}
    b_{x,j}-h_{\cS_t}(F_{x,j}^\top),\qquad
    b_{u,j}-h_{K\cS_t}(F_{u,j}^\top),
\end{align}
where $h_S(a):=\sup_{s\in S}a^\top s$ is the support function. A natural tube-size objective is therefore
\begin{equation}
\begin{aligned}
    \Phi_{\rm MPC}(\theta)=&
    \sum_{t=1}^H\sum_j \lambda_{jt} h_{\cS_t(\theta)}(F_{x,j}^\top)\\
    &+\sum_{t=0}^{H-1}\sum_j \mu_{jt} h_{K\cS_t(\theta)}(F_{u,j}^\top).
\end{aligned}
\label{eq:phi_mpc}
\end{equation}
with nonnegative weights. Other monotone objectives can be used, including terminal-set volume loss and average closed-loop cost over design states.

\subsection{Dominance over modular discard allocation}
Consider rectangular score tubes of the form
\begin{equation}
    \cE_q=\{\mathbf e:s_h(\mathbf e)\le q_h,\;h=1,\ldots,M\},
    \label{eq:rect_scores}
\end{equation}
where $h$ indexes all time, state, input, obstacle, and facet scores. A modular method chooses $q_h$ separately and may discard $r_h$ calibration trajectories in score $h$. Define
\begin{equation}
    D_h(q_h):=\{i:s_h(\mathbf E^i)>q_h\},
    \qquad |D_h(q_h)|\le r_h.
\end{equation}
The trajectory-level discard set of the same tube is
\begin{equation}
    D_\cup(q):=\bigcup_{h=1}^M D_h(q_h).
\end{equation}

\begin{proposition}[Collective dominance]\label{prop:dominance}
Let $q^{\rm mod}$ be any modular tube satisfying $|D_h(q_h^{\rm mod})|\le r_h$ for all $h$. Then
\begin{equation}
    |D_\cup(q^{\rm mod})|\le \sum_{h=1}^M r_h.
\end{equation}
Consequently, $q^{\rm mod}$ is feasible for the collective discard problem
\begin{equation}
    \min_q \Phi_{\rm MPC}(q)
    \quad \text{s.t.}\quad
    \sum_{i=1}^N\one\{\mathbf E^i\notin\cE_q\}\le r_\cup
    \label{eq:collective_rect}
\end{equation}
whenever $r_\cup\ge\sum_h r_h$. Hence the collective optimum has no larger value of any shared monotone tightening objective $\Phi_{\rm MPC}$.
\end{proposition}

\begin{proof}
A trajectory violates \eqref{eq:rect_scores} if and only if it violates at least one score threshold. Thus the set of violating trajectory indices is the union $D_\cup(q)$. Its cardinality is no larger than the sum of the individual cardinalities. Therefore the modular tube is feasible for \eqref{eq:collective_rect} with aggregate discard budget $r_\cup$. Minimization over a feasible set containing the modular tube cannot yield a worse objective value.
\end{proof}

Proposition~\ref{prop:dominance} is an optimization statement, independent of probability. The probabilistic advantage is that the collective tube is certified as one joint event; it does not spend risk separately across $M$ blocks.

The proposition should not be read as saying that any collectively optimized rectangular tube automatically has certificate complexity \(r_\cup+1\).  That equality is true for a scalar shared score with a fixed order statistic.  For a vector threshold \(q\), or for a tube shape chosen by optimization, the collective certificate must use the support/compression complexity of the complete design rule.  The practical advantage of collective allocation appears when dependence causes the same residual trajectories to be extreme in several blocks, or when a shared score/tube template deliberately exploits that dependence with a small certified complexity.

\subsection{Certified support size is the relevant complexity}

The probability certificate in Theorem~\ref{thm:beta} requires a support size that is fixed
before the calibration sample is interpreted. Thus, in the rest of the paper,
\(s\) should be understood as a certified complexity parameter: either the
exact fixed boundary size of the tube calibrator, or a conservative upper bound
on the number of calibration trajectories that can define the tube.

For example, in a shared-trajectory order-statistic tube, the same residual
trajectory may be extreme at several prediction times or constraint facets.
This overlap is useful for tube design, because it allows the collective tube
to adapt to the dependence structure of the trajectory errors. However, unless
one has a separate theorem for random boundary sizes, the observed number of
active trajectories in a particular calibration sample should {\em not} be inserted
directly into the beta-binomial formula.\footnote{Deriving exact bounds based on the observed random boundary size requires knowledge of the boundary size's probability mass function over augmented samples, which is generally intractable for complex tubes. Hence, we rely on a fixed $s$ or structural upper bounds.}

Accordingly, this paper uses the fixed-\(s\) certificate of Theorem~\ref{thm:beta}. If a
tube designer has a variable or heuristic support set, one should either certify
it with a conservative fixed upper bound on \(s\), or use the stable-compression
fallback in Appendix~\ref{app:compression} . This keeps the statistical certificate separate from
the MPC arguments: once a tube has been certified, the recursive-feasibility
and stability proofs depend only on the resulting tube inclusions.

\section{Synthesis of Shift-Compatible Calibrated Tubes}\label{sec:synthesis}
The previous sections are stated in terms of set inclusions. This section shows how those inclusions can be enforced in a finite-dimensional tube design. The purpose is to make the proposed controller implementable with the same online complexity class as standard tube MPC.

\subsection{Polytopic cross-sections}
Let the tube cross-sections have a common polytopic template
\begin{equation}
    \cS_t(\alpha_t)=\{e:G e\le \alpha_t\},
    \qquad t=1,\ldots,H,
    \label{eq:poly_tube}
\end{equation}
where the rows of $G$ are fixed and the right-hand sides $\alpha_t$ are calibrated. For a vector $a$, the support function $h_{\cS_t(\alpha_t)}(a)$ is the value of a small linear program. If the directions in $G$ include the relevant constraint normals and their images under $A_K$, then the following test is exact for the template and conservative otherwise.

\begin{proposition}[Support-function shift test]\label{prop:support_shift}
The tube \eqref{eq:poly_tube} is shift-compatible if, for every $j=0,\ldots,H-1$ and every row $g_\ell^\top$ of $G$,
\begin{equation}
    h_{\cS_1}( (A_K^j)^\top g_\ell )
    + h_{\cS_j}(g_\ell)
    \le \alpha_{j+1,\ell},
    \label{eq:support_shift}
\end{equation}
where $\cS_0=\{0\}$ and $h_{\cS_0}\equiv0$.
\end{proposition}

\begin{proof}
For any compact convex sets $S_1,S_j$ and any vector $g$, the support function of $A_K^jS_1\oplus S_j$ in direction $g$ is
\begin{equation}
    h_{A_K^jS_1\oplus S_j}(g)
    =h_{S_1}((A_K^j)^\top g)+h_{S_j}(g).
\end{equation}
Thus \eqref{eq:support_shift} says that the support of $A_K^j\cS_1\oplus\cS_j$ in every template facet direction is no larger than the corresponding support of $\cS_{j+1}$. This is precisely the inclusion $A_K^j\cS_1\oplus\cS_j\subseteq\cS_{j+1}$ when the template facets describe $\cS_{j+1}$.
\end{proof}

For axis-aligned boxes, ellipsoids with fixed shape, zonotopes with fixed generators, or homothetic polytopes, \eqref{eq:support_shift} reduces to linear or second-order cone inequalities in the tube scaling variables. Hence shift compatibility can be imposed during calibration rather than checked after the fact.

\subsection{Offline calibration as a mixed-integer tube program}
For the polytopic tube \eqref{eq:poly_tube}, a residual trajectory $\mathbf E^i$ is accepted if
\begin{equation}
    G E_t^i\le \alpha_t,
    \qquad t=1,\ldots,H.
\end{equation}
With a discard budget $r$, the collective design \eqref{eq:collective_tube_design} can be written as
\begin{subequations}\label{eq:milp_design}
\begin{align}
    \min_{\alpha,d}\quad & \Phi_{\rm MPC}(\alpha)\label{eq:milp_obj}\\
    \text{s.t.}\quad
        &G E_t^i\le \alpha_t+C d_i\mathbf 1,
        &&\forall i,t,\label{eq:milp_accept}\\
        &\sum_{i=1}^N d_i\le r,
        \qquad d_i\in\{0,1\},\label{eq:milp_discard}\\
        &\text{support-shift constraints }\eqref{eq:support_shift},\label{eq:milp_shift}\\
        &\text{terminal inclusions }\eqref{eq:terminal_ass}.\label{eq:milp_terminal}
\end{align}
\end{subequations}
Here $C$ is a sufficiently large constant. Without discards, \eqref{eq:milp_design} is a convex program for common tube templates. With discards, it is a mixed-integer program solved offline; its online consequence remains the convex MPC problem \eqref{eq:ctmpc}. This offline-online split is important, since dependence-aware calibration need not increase the real-time MPC complexity.

\begin{remark}[Scalability and Practical Synthesis] 
The MILP formulation \eqref{eq:milp_design} introduces $N$ binary variables, and its
worst-case  complexity is of  combinatorial order. For this reason, 
problem \eqref{eq:milp_design}  is best viewed as an offline shape-optimization tool only for small/moderate \((N,H,r)\), rather than a mandatory calibration step. 
 For large-scale calibration datasets (such as the $N=2500$ used in Section~\ref{sec:numerics}), the scalar shared-trajectory score tube is the recommended practical synthesis tool. It bypasses integer optimization entirely: evaluating $N$ trajectory scores requires $\mathcal{O}(N H n_g)$ operations, and the calibration threshold is found via a highly scalable $\mathcal{O}(N \log N)$ sorting operation.
 More generally, problem \eqref{eq:milp_design}  can be replaced 
 by greedy/relaxed discarding approaches. If the approach 
 possesses a ``proper boundary,'' then we can certify it via Theorem~\ref{thm:beta}.
 If it  has a stable compression set but not a proper boundary, then the fallback certificate in Appendix~\ref{app:compression} can replace Theorem~\ref{thm:beta}. In all cases,  the MPC feasibility and stability theorems remain unchanged once the resulting tube is certified.
\end{remark}

\subsection{Implementation template}
A CT-MPC implementation consists of the following operations.
\begin{enumerate}
\item Select $K$, a nominal cost, terminal ingredients, and a tube template whose projections can be used in Pontryagin differences.
\item Generate held-out residual trajectories under the same predictor and information pattern used online.
\item Solve \eqref{eq:collective_tube_design} or \eqref{eq:milp_design}, enforcing shift compatibility and terminal inclusions.
\item Compute the collective support size $s$ or a conservative upper bound, and evaluate $\eps_{N,s}(\bbeta)$.
\item At each sampling time, solve the deterministic tightened MPC problem \eqref{eq:ctmpc}; apply the first input; repeat.
\end{enumerate}
The online step is therefore indistinguishable from a robust tube MPC problem with data-designed cross-sections. The difference is in the offline selection and certification of the tube.

\section{Numerical Study}\label{sec:numerics}

This section contains two tests.  First, we measure the empirical probability that a fresh prediction-error trajectory leaves the calibrated tube.  Second, we run the resulting tightened controllers in closed loop and compare the deterministic feasible domains.  All tests, numerical tables, and graphics reported in this manuscript are generated by the accompanying reproducibility code available  at \url{https://github.com/beppe969/CT_MPC_code}


\subsection{Benchmark system, disturbance law, and controllers}

The linear benchmark is the sampled double integrator

\begin{equation}
    A=\begin{bmatrix}1&1\\0&1\end{bmatrix},
    \qquad
    B=\begin{bmatrix}0.5\\1\end{bmatrix},
\end{equation}

with hard constraints $|p|\le 5$, $|v|\le2.5$, and $|u|\le1.2$.  The ancillary gain is

\begin{equation}
    K=\begin{bmatrix}-0.593&-1.121\end{bmatrix},
\end{equation}
for which $A+BK$ is Schur.  The prediction horizon is $H=15$.  The online MPC problem is a deterministic tightened tube-MPC problem with a tube-dependent polyhedral terminal set.  In the reproducibility code this set is generated from the terminal gain $K_f=[-1\;-1.5]^\top$ by the same linear support-function inclusions used in Assumption~\ref{ass:terminal}.
  The online objective is an $\ell_1$ tracking objective.  This objective was chosen to keep the replication code independent of commercial LP solvers beyond SciPy/HiGHS; the tightening and feasibility comparisons do not rely on this particular cost.

The disturbance generator is non-Gaussian and temporally dependent.  For a prediction horizon, let

\begin{equation}
    \xi_t=\sqrt{\rho}\,\xi_c+\sqrt{1-\rho}\,\xi_t^0,
\end{equation}

where $\xi_c$ and $\xi_t^0$ are independent Student-$t$ variables with four degrees of freedom, normalized to unit variance.  The disturbance is

\begin{equation}
    w_t=\begin{bmatrix}0.075\\0.045\end{bmatrix}\xi_t+
    \begin{bmatrix}0.020\eta_{p,t}\\0.024\eta_{v,t}\end{bmatrix},
    \label{eq:numerical_disturbance}
\end{equation}

where $\eta_{p,t}$ and $\eta_{v,t}$ are independent standard Gaussian variables.  The scalar $\rho\in[0,1]$ controls temporal dependence.  Each reported design uses $2500$ independent tuning trajectories for score geometry, $2500$ independent calibration trajectories for the finite-sample certificate, and $8\cdot10^4$ independent test trajectories for the empirical risk estimates.  At risks near $0.05$, the Monte Carlo standard error is at most $7.7\cdot10^{-4}$.

We compare four tubes.  The collective tube uses tuning-split scales

\begin{equation}
    \sigma_{t,j}=Q_{0.75}\{|\widetilde E^\nu_{t,j}|:\nu\in\cI_{\rm tune}\}
\end{equation}

and calibration scores

\begin{equation}
    S^i=\max_{t=1,\ldots,H}\max_{j\in\{p,v\}}
    \frac{|E^i_{t,j}|}{\sigma_{t,j}}.
\end{equation}
We set $r=100$ and choose $q$ as the $(r+1)$-st largest calibration score, equivalently the noninterpolated upper empirical quantile that accepts $N-r=2400$ calibration trajectories.  Thus $s=r+1=101$.  
With $N=2500$ and confidence parameter $\beta=0.05$, the beta-binomial certificate gives $\eps_{2500,101}(0.05)=0.0471<0.05$.

As a stronger recent baseline, we implemented the joint-in-time conformal ellipsoidal construction used in conformal SMPC for stochastic linear systems in \cite{Vogel2026}.  On the tuning split, it estimates $\mu_t$ and $\Sigma_t$ for each prediction time.  On the calibration split, it uses the trajectory score

\begin{equation}
    S^i_{\rm ell}=\max_{t=1,\ldots,H}
    (E^i_t-\mu_t)^\top\Sigma_t^{-1}(E^i_t-\mu_t),
\end{equation}
calibrated at level $1-\eps$.  The resulting ellipsoid is then projected to the symmetric box half-widths $|\mu_{t,j}|+\sqrt{q_{\rm ell}[\Sigma_t]_{jj}}$ so that it can be used by the same deterministic tube-MPC solver.  This baseline calibrates a joint trajectory event and is therefore substantially stronger than a marginal Bonferroni allocation.  The third tube is the Bonferroni tube $Q_{1-\eps/(2H)}\{|E^i_{t,j}|\}$ with $\eps=0.05$.  The fourth is the component-wise sample envelope, included only as a robust-conservatism reference.
%

Before solving the collective MPC problem, the raw score box $b^{\rm raw}_{t,j}=q\sigma_{t,j}$ is replaced by its smallest componentwise shift-compatible box closure.  For symmetric boxes this closure keeps $\bar b_1=b^{\rm raw}_1$ and applies
\begin{equation}
    \bar b_{j+1}=\max\{b^{\rm raw}_{j+1},\ \bar b_j+|A_K^j|\bar b_1\},
    \qquad j=1,\ldots,H-1,
    \label{eq:numerical_shift_closure}
\end{equation}
componentwise.  Since the closed tube contains the raw calibrated tube, the same trajectory-level risk certificate remains valid for the MPC tube.  The open-loop risk table reports the calibrated raw score tubes, while the closed-loop MPC and feasible-domain computations use the closed collective tube.

With tolerance $10^{-8}$, the closed collective tubes have maximum support-shift residual $0$ for all tested values of $\rho$; at $\rho=0.85$ the unclosed score box would have residual $1.918\cdot10^{-1}$ and the largest added half-width in the closure is $2.999\cdot10^{-1}$.  The terminal polytope generated with $K_f=[-1\;-1.5]^\top$ satisfies the terminal safety, shifted-state, shifted-input, and invariance inclusions with nonpositive residuals to numerical precision; at $\rho=0.85$ the remaining terminal input margin is $2.931\cdot10^{-1}$.  

\subsection{Open-loop tube calibration test}

Table~\ref{tab:extended_open_loop} reports the empirical risks of the three calibrated comparison tubes and also lists the component-wise envelope risk as a robust-conservatism reference; Figure~\ref{fig:dependence_sweep} plots all four risk curves.  The collective tube remains below the target joint risk in all tested dependence regimes.  The joint conformal ellipsoid baseline is also safe and much stronger than the sample envelope, but its box projection is larger than the collective box tube.  At $\rho=0.85$, the raw collective calibration reduces the sum of state-error half-widths by $35.5\%$ and the mean rectangular cross-section area by $54.6\%$ relative to Bonferroni; the joint conformal ellipsoid reduces them by $20.6\%$ and $33.3\%$, respectively.

\begin{table}[t]
\centering
\begingroup
\caption{Open-loop tube calibration test. Risks are empirical probabilities that some component of the horizon error trajectory leaves the designed box tube. The envelope column is included only as a robust-conservatism reference. Reductions are for the collective tube relative to the Bonferroni tube.}
\label{tab:extended_open_loop}
\setlength{\tabcolsep}{3pt}
\begin{tabular}{ccccccc}
\toprule
$\rho$ & Coll. risk & JCP risk & Bonf. risk & {Env. risk} & Sum red. & Area red.\\
\midrule
0.00 & 0.0467 & 0.0261 & 0.0254 & {0.0059} & 17.2\% & 31.3\%\\
0.30 & 0.0426 & 0.0195 & 0.0176 & {0.0055} & 20.2\% & 35.8\%\\
0.60 & 0.0392 & 0.0151 & 0.0114 & {0.0018} & 27.1\% & 45.1\%\\
0.85 & 0.0388 & 0.0162 & 0.0077 & {0.0039} & 35.5\% & 54.6\%\\
0.95 & 0.0409 & 0.0195 & 0.0060 & {0.0009} & 45.4\% & 66.3\%\\
\bottomrule
\end{tabular}
\endgroup
\end{table}

\begin{figure*}[t]
\centering
\includegraphics[width=0.92\textwidth]{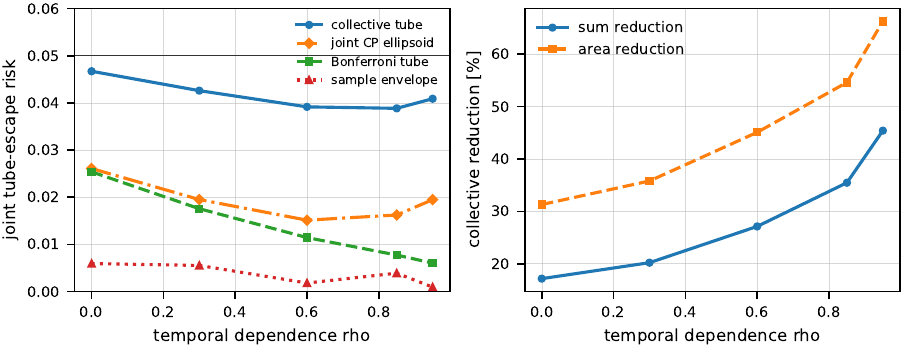}
\caption{Open-loop tube test. Left: empirical joint tube-escape probability for the collective tube, joint conformal ellipsoid (JCP) baseline, Bonferroni tube, and sample envelope; the horizontal line is the target joint risk $0.05$. Right: reduction of the collective tube relative to the Bonferroni tube.}
\label{fig:dependence_sweep}
\end{figure*}

Figure~\ref{fig:tube_profiles} shows the raw calibrated tube cross-sections for $\rho=0.85$, before the collective shift-compatible closure used by the MPC solver. Here \(b_p(t)\) and \(b_v(t)\) denote the position- and velocity-error half-widths of the raw box cross-section at prediction step \(t\), respectively. The lower panel translates each raw tube into the tightened input slack $1.2-\max_{e\in\cS_t}|Ke|$. 
%
The joint conformal ellipsoid baseline is competitive in closed loop, but its ellipsoidal confidence region becomes a conservative box when projected coordinatewise.

\begin{figure}[t]
\centering
\includegraphics[width=\linewidth]{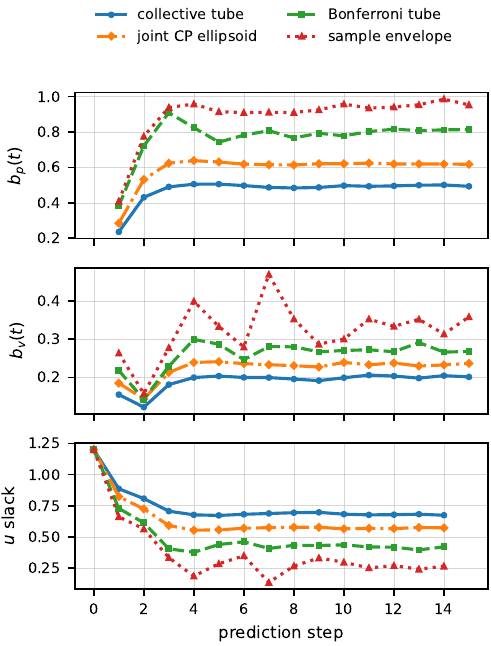}
\caption{Raw calibrated tube geometry for $\rho=0.85$, before applying the collective shift-compatible closure used by the MPC solver.  The collective raw score tube gives the smallest aggregate box tightening and the largest raw input slack after the ancillary feedback margin $|K|b_t$ is removed.}
\label{fig:tube_profiles}
\end{figure}

\subsection{Closed-loop MPC test and feasible-domain enlargement}

We run the receding-horizon controller from $x_0=(4.8,0.3)$ for $35$ sampling instants under \eqref{eq:numerical_disturbance} with $\rho=0.85$.  The collective, joint conformal ellipsoid, and Bonferroni controllers are tested on $40$ closed-loop rollouts each; the sample-envelope controller is tested on $20$ rollouts.  Table~\ref{tab:closed_loop_extended} reports zero recursive-feasibility failures and zero hard state/input violations.  With the admissible collective closure, the collective and joint conformal controllers have nearly identical costs, while the Bonferroni and envelope tubes remain more conservative.

\begin{table}[t]
\centering
\begingroup
\caption{Closed-loop MPC test at $\rho=0.85$ from $x_0=(4.8,0.3)$ over $35$ sampling instants. ``Fail'' is the observed recursive-feasibility failure rate. ``Viol.'' is the observed hard state/input violation rate.}
\label{tab:closed_loop_extended}
\begin{tabular}{lccccc}
\toprule
Tube & Runs & Cost & $\Delta$ cost & Fail & Viol.\\
\midrule
Collective & 40 & 56.89 & 0.00\% & 0 & 0\\
Joint CP & 40 & 56.83 & $-0.11\%$ & 0 & 0\\
Bonferroni & 40 & 57.29 & 0.70\% & 0 & 0\\
Envelope & 20 & 58.77 & 3.31\% & 0 & 0\\
\bottomrule
\end{tabular}
\endgroup
\end{table}

Figure~\ref{fig:feasible_region} evaluates the deterministic tightened MPC problem on a grid of $475$ initial states in the hard constraint box.  The collective controller, using the shift-compatible closure, is feasible at $407$ grid points, while the Bonferroni tube is feasible at $395$ grid points.  The $12$ additional grid points are near the constraint boundary.  

\begin{figure}[t]
\centering
\includegraphics[width=\linewidth]{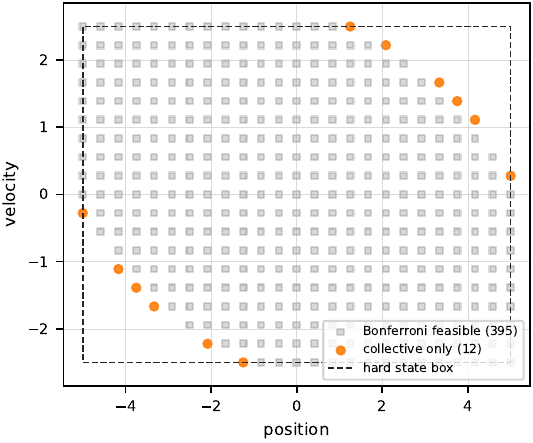}
\caption{Tightened MPC feasible region at $\rho=0.85$ using the admissible collective closure.  Gray squares are grid points feasible for the Bonferroni tube.  Orange circles are additional grid points feasible for the collective tube but infeasible under Bonferroni tightening.  The dashed rectangle is the original hard state box.}
\label{fig:feasible_region}
\end{figure}

\section{Discussion}\label{sec:discussion}
The proposed framework separates three concepts that are often conflated.  First, \emph{trajectory calibration} is an offline operation: exchangeable residual trajectories select a tube and certify the probability that a fresh trajectory leaves it.  No distributional model for disturbances is needed, but the exchangeability condition must be justified by the data-collection protocol.  Second, \emph{tube admissibility} is deterministic, since shift compatibility, Pontryagin tightening, and terminal invariance are set inclusions.  They can be checked or enforced exactly for polytopic and ellipsoidal tube families, and approximately for learned tube parameterizations.  Third, \emph{closed-loop repeated use} is a separate stochastic problem, for which  Theorems~\ref{thm:horizon_safety} and \ref{thm:rf} give a calibrated one-horizon and one-step recursive-feasibility certificate.  A finite deployment is handled by \eqref{eq:T_bound}; indefinite hard safety requires additional structure such as online calibration updates,   or a deterministic robust backup tube.

The baselines in the numerical tests Section~\ref{sec:numerics}, including the joint conformal ellipsoid baseline, are not meant to exhaust the (vast) stochastic-MPC literature. 
 Advanced randomized MPC methods can also obtain joint probabilistic statements without allocating risk to every time-facet separately.  Our distinction, however, is architectural, since  those methods typically certify a sampled optimization problem through scenario support/rank arguments, whereas CT-MPC certifies an offline tube object and then uses deterministic tube-MPC inclusions online.  The numerical comparison tests the specific question addressed here: whether calibrating the whole residual trajectory as the tube object gives less conservative box tightenings than marginal allocation or ellipsoid-to-box projection while preserving the standard recursive-feasibility proof.

The present theory is focused on linear systems with additive uncertainty and convex constraints because this setting cleanly separates the statistical certificate from the deterministic MPC proof.  Nonlinear, output-feedback, and adaptive variants are natural next steps, but they should be treated as integrated problems, since the calibration layer can certify a trajectory tube only after the chosen MPC layer supplies the required shifted candidate, robust tightenings, and terminal invariance conditions.

\section{Conclusion}\label{sec:conclusion}
This paper reformulated data-calibrated uncertainty quantification for MPC around the object that MPC actually uses, that is, a finite-horizon tube.  Rather than calibrating many marginal risks and combining them with a union bound, CT-MPC calibrates a single trajectory-level error tube, uses its projections for deterministic Pontryagin tightening, and certifies the probability that the complete prediction-error trajectory leaves the tube.  For linear additive systems with fixed ancillary feedback, the resulting controller preserves the core tube-MPC structure, in which tube containment implies joint state-input safety, recursive feasibility follows by a shifted candidate, and the value function satisfies a practical decrease inequality.  The finite-sample certificate is distribution-free and depends on collective support complexity, i.e., the number of residual trajectories defining the tube.  This provides a control-theoretic route to non-Bonferroni stochastic MPC with explicit dependence awareness.

For deployments that require indefinite hard safety, the natural implementation is a hybrid one in which the collective tube is used as the performance tube, while a robust invariant backup tube with a one-step guard provides a deterministic fallback whenever calibrated containment is lost.  This separates finite-sample distribution-free performance certification from infinite-horizon safety invariance.

\appendices

\section{Boundary Verification for Two Tube Families}\label{app:boundary}
For shared-trajectory order-statistic tubes, let $s_h$ be a score and let the threshold in score $h$ be determined by the $r_h+1$ largest calibration score values. In any augmented sample, the boundary for score $h$ is the set of those $r_h+1$ trajectories. The collective boundary is their union over $h$. If omitted trajectories are all accepted by the tube, then none of the full-sample boundary trajectories can have been omitted; otherwise the corresponding order statistic would be exceeded. Conversely, if the boundary union is retained, all omitted trajectories are below every retained order statistic and are accepted. Projectivity follows because retaining the boundary preserves the same order-statistic-defining trajectories.

For the Pareto lower-orthant tube, a trajectory is in the boundary if its score vector is Pareto maximal. If an omitted Pareto-maximal score vector were accepted by the envelope of retained scores, it would be dominated by a retained score vector, contradicting maximality. Conversely, if all Pareto maxima are retained, every nonmaximal omitted score vector is dominated along a finite chain by a retained Pareto maximum, and hence is accepted. Removing nonmaximal points cannot create new Pareto maxima when all full-sample maxima are retained, proving projectivity.

\section{Fallback Certificate for Irregular Tube Designers}\label{app:compression}
The sharp beta law requires boundary equivalence. Some useful tube designers, such as greedy discard rules or learned tube-shape selection, may not satisfy it. A conservative fallback can be obtained from stable sample compression. Suppose a tube designer returns $\hat\theta_N$, a compression set $C_N\subset[N]$, and a discarded set $D_N\subset[N]$. Let $b_N:=|C_N\cup D_N|$ and
\begin{equation}
    \widehat V_N:=\frac{1}{N}\sum_{i=1}^N \one\{\mathbf E^i\notin\cE_{\hat\theta_N}\}.
\end{equation}
Assume that $\hat\theta_N$ can be reconstructed from the samples indexed by $C_N\cup D_N$, that removing any noncompressed accepted sample does not change the reconstructed tube, and that all non-discarded calibration samples are accepted. Then the stable-compression bound of \cite{Hanneke2021} applied to the binary loss $\one\{\mathbf e\notin\cE_\theta\}$ yields, with probability at least $1-\bbeta$,
\begin{equation}
\begin{aligned}
V_N \le& \widehat V_N+
\sqrt{\frac{72\widehat V_N}{N}\left(2b_N+\log\frac{4e}{\bbeta}\right)}\\
&+\frac{32}{N}\left(2b_N+\log\frac{4e}{\bbeta}\right).
\end{aligned}
\label{eq:compression_bound}
\end{equation}
Since $\widehat V_N\le |D_N|/N$, \eqref{eq:compression_bound} gives a computable certificate for irregular tube designers. It is looser than Theorem~\ref{thm:beta}, but it preserves the separation between tube certification and MPC recursive-feasibility analysis.

\bibliographystyle{IEEEtran}
\bibliography{references}

\begin{IEEEbiography}[{\includegraphics[width=1in,height=1.25in,clip,keepaspectratio]{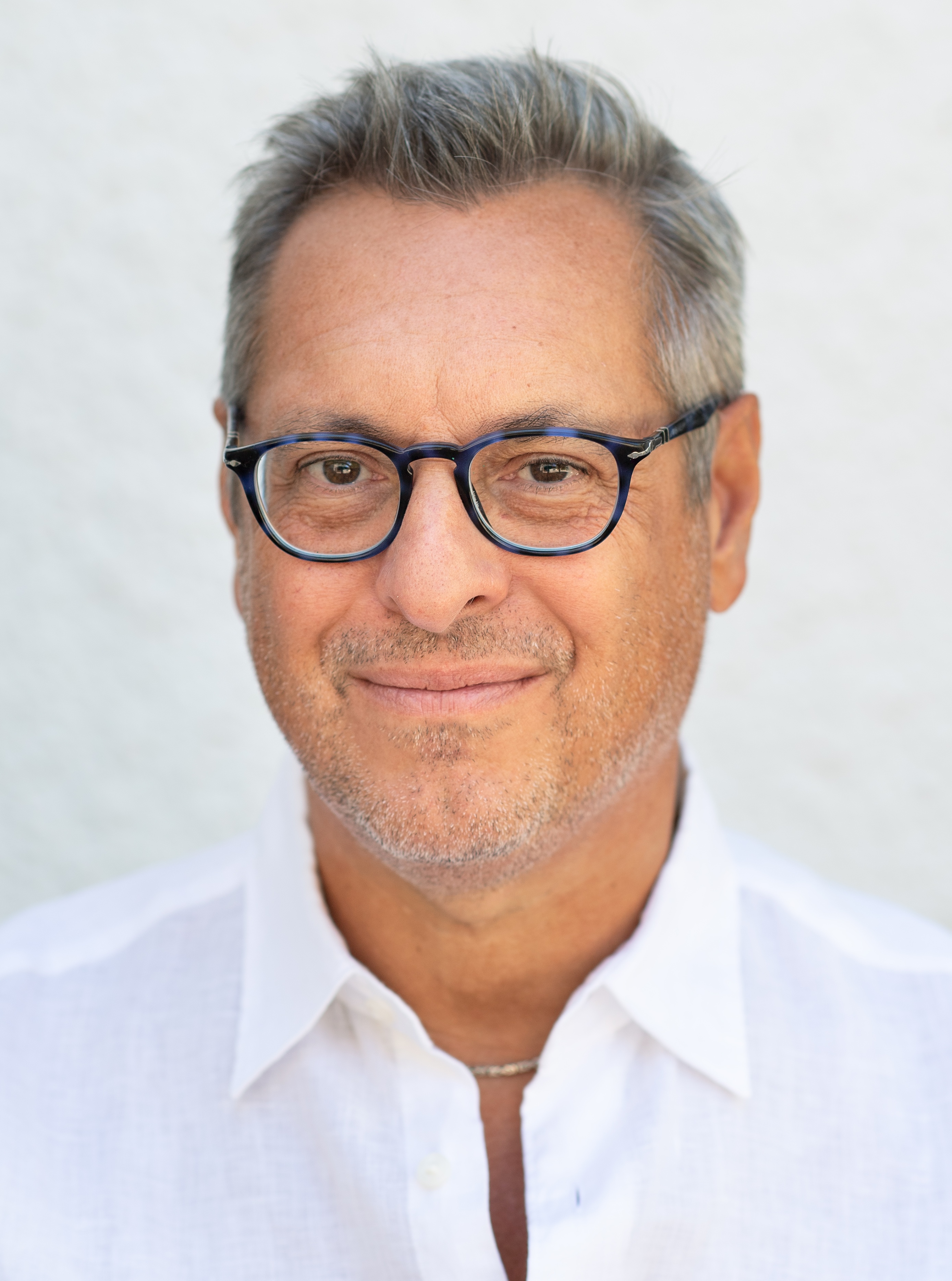}}]{Giuseppe C. Calafiore}
 is a full professor at DET, Politecnico di Torino, where he coordinates the Control Systems and Data Science group. He was an associate fellow of the IEIIT-CNR, Italy, and a Professor of Mechanical Engineering at CECS, VinUniversity, Vietnam, in 2023. Dr. Calafiore was a visiting scholar at the Information Systems Laboratory (ISL), Stanford University, California, in 1995; at the Ecole Nationale Sup\'erieure de Techniques Avanc\'ees (ENSTA), Paris, in 1998; and at the University of California at Berkeley, in 1999, 2003, 2007, 2017–19, and 2021 where he co-taught a Master course on Financial Data Science.  He was a Senior Fellow at the Institute of Pure and Applied Mathematics (IPAM), UCLA, in 2010. Dr. Calafiore is the author of about 220 journal and conference
proceedings papers, and of nine books.
He is a Fellow of the IEEE (2018), a Fellow of the Asia-Pacific Artificial Intelligence Association (AAIA), a Fellow
of the International Artificial Intelligence Industry Alliance (AIIA), and a Fellow Member of the European Academy of Sciences (EurASc). Dr. Calafiore received the IEEE Control System Society “George S. Axelby” Outstanding Paper Award in 2008. His research interests are in the fields of convex optimization, data science, machine learning, and identification and control of uncertain systems, with applications ranging from computational finance to networked systems, power management and smart mobility.\end{IEEEbiography}

\end{document}